\documentclass[10pt, doublecolumn]{IEEEtran}
\usepackage{epsfig,latexsym}
\usepackage{float}
\usepackage{indentfirst}
\usepackage{amsmath}
\usepackage{bm}
\usepackage{amssymb}
\usepackage{times}
\usepackage{enumitem}

\usepackage{algorithm}
\usepackage[noend]{algpseudocode}
\usepackage{subfigure}
\usepackage{psfrag}
\usepackage{hyperref}
\usepackage{cite}
\usepackage{lastpage}
\usepackage{fancyhdr}
\usepackage{color} 
 \usepackage{amsthm}
\usepackage{bigints}
\usepackage{array}
\usepackage{booktabs}
\newtheorem{Lemma}{Lemma}
\newtheorem{Corollary}{Corollary}
\newtheorem{lemma}[Lemma]{$\mathbf{Lemma}$}

\newtheorem{corollary}[Corollary]{$\mathbf{Corollary}$}

\newcounter{problem}
\newcounter{save@equation}
\newcounter{save@problem}
\makeatletter

\begin{document}
\title{  \vspace{-1em}{\LARGE Characterizing Multi-Cell Pinching-Antenna Transmission: \\Revealing the Other Side of the Coin   }}

\author{ Zhiguo Ding, \IEEEmembership{Fellow, IEEE}  \thanks{ 
  
\vspace{-1em}

Z. Ding is with the School of Electrical and Electronic Engineering, Nanyang Technological University, Singapore.  
 

  }\vspace{-3em}}
 \maketitle

\begin{abstract}
Using pinching antennas to enhance a user's connection to its own base station (BS) is intuitive and has been well investigated in the literature. This letter focuses on a less intuitive advantage of pinching antennas from the interference-suppression perspective. In particular, using pinching antennas ensures that all BSs serve their users with low transmit power, i.e., BSs can ``whisper", rather than ``shout", to their users. As a result, by experiencing less interference, a user's data rate can still be improved with pinching antennas, even if its link quality to its own BS remains unchanged.  A stochastic geometric study is carried out in the letter, where  BSs employ fractional power control. Analytical and numerical results are presented to reveal the significant impact of this interference-suppression capability offered by pinching antennas.
\end{abstract}\vspace{-0.5em}

\vspace{-1em} 

\section{Introduction}
As a key enabler for line-of-sight (LoS) communications, pinching antennas have recently received significant attention \cite{pinching_antenna1,mypa}. The key idea of pinching antennas is to utilize the flexible radio-frequency emission feature of dielectric waveguides, which ensures that base stations (BSs) can be deployed in close proximity to their serving users. As a result, the use of pinching antennas can significantly improve a user's connection to its own BS, and the performance gain due to this link quality improvement has been well studied in the literature on pinching antennas \cite{11414084,11348983,11434944}.

This letter focuses on a less intuitive advantage of pinching antennas from the interference-suppression perspective. Recall that using pinching antennas allows all BSs to serve their users with lower transmit power than the conventional-antenna-based system. In particular, in a conventional multi-cell system where each BS is placed at the center of its cell, BSs have to use high transmit power to reach their users, particularly for the users close to the cell edge, i.e., BSs have to ``shout" at their serving users. Pinching antennas can be flexibly deployed in close proximity to the users, and hence BSs can reduce their transmit powers, i.e., BSs ``whisper" to their users. As a result, by experiencing less interference, a user's data rate can still be improved with pinching antennas, even if its link quality to its own BS remains unchanged. A stochastic geometric study is carried out  to characterize the superior performance gains enabled by this interference-suppression capability of pinching antennas. Unlike the existing stochastic geometric works in \cite{11195162,11315149}, BSs are assumed to employ fractional power control, i.e., a BS's transmit power is dynamically adjusted according to its serving user's channel condition.  Both analytical and numerical results are presented to reveal the significant impact of the interference-suppression capability offered by pinching antennas.

\section{System Model}\label{section system model}
Consider a multi-cell downlink communication scenario, where BSs follow a homogeneous Poisson point process (HPPP)    with its intensity denoted by $\lambda$ \cite{Haenggi}. Each BS covers a disc-shaped cell with radius $r_{\rm c}$ and is equipped with a dielectric waveguide with length of $2r_{\rm c}$. Each BS serves its users in a time-division multiple access (TDMA) manner, e.g., in each time slot, a single pinching antenna is activated to serve a single user.  

Denote the BSs and their serving users by ${\rm BS}_i$ and ${\rm U}_i$, $i\geq 1$, respectively. By using Slivnyak’s theorem, an additional BS at the origin can be added to the HPPP without affecting the distribution of the remaining points in the process \cite{Haenggi}. This BS and its user are treated as the typical BS and user, and denoted by ${\rm BS}_0$ and ${\rm U}_0$, respectively. Denote the distance between ${\rm BS}_i$ and ${\rm U}_i$  by $R_i$, and the distance between ${\rm BS}_i$ and ${\rm U}_0$ by $D_i$, $i\neq 0$.

Existing stochastic geometric studies in  \cite{11195162,11315149} assumed that BSs relied on fixed transmit power, which is not suitable for characterizing interference suppression. Instead, fractional power control is employed in this letter, i.e., $P_i = P R_i^{ \alpha\epsilon}$, where $P_i$ denotes ${\rm BS}_i$'s transmit power, $P$ denotes the ideal transmit power, and $\alpha$ denotes the path loss exponent, e.g., $\alpha=2$ for LoS links and $\alpha=4$ for non-LoS (NLoS) links \cite{7506127}.

Therefore, the downlink data rate achievable at the typical user ${\rm U}_0$ is given by 
\begin{align}\label{ssm}
R_0^{\rm PA} = \log\left(1+\frac{  P R_0^{\alpha\epsilon} g_0}{\sum^{\infty}_{i=1}  P R_i^{\alpha\epsilon}  g_i +P_{\rm N}}\right),
\end{align}
where $P_{\rm N}$ denotes noise power, and the channel gain $g_i$, $0\leq i \leq M$, is defined as follows. 


If an LoS link is available, only the LoS component is considered, since it dominates the NLoS component. Therefore, for the LoS case, $\alpha=2$, $g_i=\frac{\eta }{D_i^2}$ for $i\geq 1$, and $g_0=\frac{\eta }{R_0^2}$, where $\eta$ denotes the free-space propagation constant \cite{mypa}. An NLoS link is assumed to be   Rayleigh faded, i.e., $\alpha=4$, $g_i=\frac{\eta h_i}{D_i^4}$ for $i\geq 1$, and $g_0=\frac{\eta h_0}{R_0^4}$, where $h_i$ is exponentially distributed, i.e., $f_h(x) = e^{-x}$. 
The probability for the existence of an LoS link is modeled as follows: $\mathbb{P}_{\rm L}=e^{-\beta D}$, where $D$ denotes the transceiver distance and $\beta$ denotes the blockage parameter \cite{11585818}.

{\it Remark 1:} Due to the flexibility feature of pinching antennas, the transceiver distance in the pinching-antenna system is smaller than that of the conventional system. Therefore, with the implementation of fractional power control, BSs' transmit powers in the conventional scheme need to be normalized to ensure a fair comparison with the pinching-antenna system, where the details about power normalization for the benchmark are omitted due to space limitations.  

\vspace{-1em}
\section{A One-Dimensional Special Case}\label{section 3}
We first focus on the one-dimensional (1D) special case, where the BSs and users are deployed along a line, as commonly encountered in vehicular and tunnel communication systems. To obtain insightful understandings, we further simplify the interference term in \eqref{ssm} by considering the strongest interference source only, i.e., the BS nearest to ${\rm BS}_0$. Denote this BS by ${\rm U}_{\rm I}$, and its coordinate is denoted by $(r,0)$. It is straightforward to verify that $|r|$ follows the probability density function (pdf):  $f_{|r|}(x) = 2\lambda e^{-2\lambda x}$, which means that the pdf of $r$ is give by  $f_{r}(r) = \lambda e^{-2\lambda |r|}$ \cite{11195162}.

The coorindates of the users of ${\rm BS}_0$ and ${\rm BS}_{\rm I}$ are  denoted by $(z_0,0)$ and $(z_{\rm I}+   r,0)$, respectively, where $z_0$ and  $z_{\rm I}$ are two  independent and identically distributed (i.i.d.) uniform random variables between $-r_{\rm c}$ and $r_{\rm c}$.  To simplify analysis, it is assumed that  LoS links always exist. For the considered special case, ${\rm U}_0$'s achievable data rate is given by 
\begin{align}
R_0^{\rm PA} = \log\left(1+ \frac{  P d^{\alpha\epsilon} \frac{\eta}{d^2}}{   P d^{\alpha\epsilon}  \frac{\eta}{(z_0-z_{\rm I}-  r)^2+d^2}  +P_{\rm N}}\right),
\end{align}
where $d$ denotes the height of the waveguide. 
 
 \vspace{-1em}
 \subsection{Outage Probability}
The outage probability achieved by pinching-antenna-assisted multi-cell transmission is given by
\begin{align}
\mathbb{P}^o =& \mathbb{P}\left(
 \log\left(1+ \frac{  P d^{\alpha\epsilon} \frac{\eta}{d^2}}{   P d^{\alpha\epsilon}  \frac{\eta}{(z_0-z_{\rm I}-    r)^2+d^2}  +P_{\rm N}}\right)\leq R
\right),
\end{align}
where $R$ denotes the target data rate. 

Define $\tau= \sqrt{ \frac{1}{        \frac{1}{d^2\left(e^R-1\right)}  -\frac{P_{\rm N}}{\eta P d^{\alpha\epsilon} } }-d^2}$, and the above outage probability can be expressed as follows: 
\begin{align}\label{po1}
\mathbb{P}^o =&  
  \mathbb{P}\left(
     -\tau \leq   z_0-z_{\rm I}-    r  \leq \tau
\right).
\end{align}

Define $z=z_0-z_{\rm I}$. Because $z_0$ and $z_{\rm I}$ are i.i.d. uniforml distributed between $-r_c$ and $r_c$, it is straightforward to show that the pdf of $z$ is given by $f_z(z) = \frac{2r_c-|z|}{4r_c^2}$, for $-2r_c\leq z\leq 2r_c$. Since   $\mathbb{P}( r\geq 0)=\mathbb{P}( r\leq0)=\frac{1}{2}$ and the pdf of $z$ is an even function, \eqref{po1} can be simplified as follows:
\begin{align}
\mathbb{P}^o =&  \mathbb{P}\left(
     -\tau\leq z-   |{r}|  \leq \tau
\right). 
\end{align}
By using the pdfs of $z$ and $|r|$ and with some straightforward algebraic manipulations, a closed-form expression for $\mathbb{P}^o$ can be obtained as in the following lemma.   

\begin{lemma}\label{lemma1}
For the considered 1-D special case, the outage probability is given by  $\mathbb{P}^o=
\frac{ \tau}{r_{\rm c}}
-\frac{ \tau^2}{4r_{\rm c}^2}
+
\frac{
2e^{-2\lambda  \tau}-2
+e^{-2\lambda(2r_{\rm c}- \tau)}
-e^{-2\lambda(2r_{\rm c}+ \tau)}
}
{16\lambda^2r_{\rm c}^2}$, if $0\le  \tau<2r_{\rm c}$. Otherwise, $\mathbb{P}^o= 
1-
\frac{
e^{-2\lambda( \tau-2r_{\rm c})}
-2e^{-2\lambda  \tau}
+e^{-2\lambda( \tau+2r_{\rm c})}
}
{16\lambda^2r_{\rm c}^2}$.  
\end{lemma}
We first note that at high SNR, i.e., $\frac{P}{P_{\rm N}}\rightarrow \infty$,   $\tau$ can be approximated as follows: 
\begin{align}
\tau=& \sqrt{ \frac{1}{        \frac{1}{d^2\left(e^R-1\right)}  -\frac{P_{\rm N}}{\eta P d^{\alpha\epsilon} } }-d^2}\approx  d\sqrt{e^R-2},
\end{align}
which means that the case of  $0\leq \tau<2r_{\mathrm{c}}$ is more likely to happen in practice. For this case, the outage probability can be approximated as follows:
\begin{align}\nonumber
&\mathbb{P}^o
   \approx \frac{ \tau}{r_{\rm c}}
-\frac{ \tau^2}{4r_{\rm c}^2}
 +
\frac{1}
{16\lambda^2r_{\rm c}^2}  \left(
2\left(1-2\lambda  \tau+2\lambda^2a^2-\frac{4}{3}\lambda^3a^3\right)\right.\\\nonumber &\left.-2
+1-2\lambda(2r_{\rm c}- \tau)+2\lambda^2(2r_{\rm c}- \tau)^2
-\frac{4}{3}\lambda^3(2r_{\rm c}- \tau)^3\right.\\ \nonumber &\left.
-1 + 2\lambda(2r_{\rm c}+ \tau)- 2\lambda^2(2r_{\rm c}+ \tau)^2+\frac{4}{3}\lambda^3(2r_{\rm c}+ \tau)^3\right),
\end{align}
if $\lambda $ is very small, 
where the approximation, $e^{-x}\approx \sum^{3}_{n=0}(-1)^n\frac{x^n}{n!}$,  is used. 
 By using the above approximation, the following corollary can be obtained. 
 \begin{corollary}\label{corollary1}
 For the considered special case with small $\lambda$, the outage probability can be approximated as follows:
 \begin{align}
&\mathbb{P}^o \approx \lambda
  d\sqrt{e^R-2} \left(2+\frac{  d^2{e^R-2}}
{6r_{\rm c}^2}\right).
\end{align}
 \end{corollary}

{\it Remark 2:} The analytical result shown in Corollary \ref{corollary1} indicates that in the small-$\lambda$ regime,  the outage probability achieved by pinching-antenna assisted multi-cell transmission decreases as  $\lambda$ decreases and the cell radius increases.

\vspace{-1em}
\subsection{An Analytical Comparison to the Considered Benchmark}
This section aims to study the performance gain offered by pinching antennas over conventional antennas. For the considered special case, the  data rate achieved by the conventional-antenna system can be expressed as follows: 
\begin{align}
R_0^{\rm Conv} = \log\left(1+ \frac{  P (z_0^2+d^2)^{ \epsilon} \frac{\eta}{z_0^2+d^2}}{   P (z_{\rm I}^2+d^2)^{ \epsilon}  \frac{\eta}{(z_0-r)^2+d^2}  +P_{\rm N}}\right).
\end{align}
To facilitate an insightful comparison, we assume that $z_0=0$ and $\epsilon=1$. We note that $z_0=0$ implies that ${\rm U}_0$ is directly beneath ${\rm BS}_0$, i.e., the use of pinching antennas does not improve the typical user's channel gain, which is aligned with the aim of this letter for identifying the interference-suppression capability of pinching antennas. 

For the considered special case, the probability for conventional antennas outperforming pinching antennas is given by   
\begin{align}\nonumber
\mathbb{P}^g \triangleq & \mathbb{P}\left(R_0^{\rm conv} \geq R_0^{\rm PA}\right) =\mathbb{P}\left(  \frac{z_{\rm I}^2+d^2}{r^2+d^2}  \leq
  \frac{d^2}{(z_{\rm I}+r)^2+d^2}  \right). 
\end{align}
 
The probability $\mathbb{P}^g $ can be simplified as follows:
\begin{align}
\mathbb{P}^g = &\mathbb{P}\left( z_{\rm I}(z_{\rm I}+r)(z_{\rm I}^2+rz_{\rm I}+2d^2)\leq 0 , z_{\rm I}\leq 0 \right) \\\nonumber
 &+\mathbb{P}\left( z_{\rm I}(z_{\rm I}+r)(z_{\rm I}^2+rz_{\rm I}+2d^2)\leq 0 ,z_{\rm I}> 0  \right) .
\end{align}
It can be verified that the two probabilistic terms in the above expression are identical. Therefore, $\mathbb{P}^g $ can be evaluated as follows:
 \begin{align} \mathbb{P}^g=
 &
 \int^{0}_{-r_c} \frac{1}{r_c}\int^{-z- \frac{2d^2}{z} }_{-z} \lambda e^{-2\lambda r}dr dz.
\end{align}
With some straightforward algebraic manipulations, the following lemma can be obtained. 
\begin{lemma}\label{lemma2}
For the considered special case, the probability for the conventional-antenna system to outperform the pinching-antenna system is given by
\begin{align}
\mathbb{P}^g&=
\frac{1-e^{-2\lambda r_{\rm c}}}
{4\lambda r_{\rm c}}\\\nonumber &
-
\frac{\sqrt{2}d}{r_{\rm c}}
\left[
K_1\left(4\sqrt{2}\lambda d\right)
-\frac{r_{\rm c}}{2\sqrt{2}d}
K_{-1}\left(
2\lambda r_{\rm c},
\frac{4\lambda d^2}{r_{\rm c}}
\right)
\right],
\end{align}
where $K_n(\cdot)$ denotes the modified Bessel function of the second kind of order $n$ and $K_n(\cdot, \cdot)$ denotes the upper incomplete modified Bessel function of the second kind \cite{besselsss}. 
\end{lemma}

Lemma \ref{lemma2} leads to the following approximation. 

\begin{lemma}\label{lemma3}
For the considered special case, if $\frac{\lambda d^2 }{r_{\rm c}}\rightarrow 0$,  the probability for the conventional-antenna system to outperform the pinching-antenna system can be approximated as follows: 
\begin{align}\label{lemma3eq}
\mathbb{P}^g \approx \frac{2\lambda d^2 }{r_{\rm c}} \left(- \log \left(\frac{4\lambda d^2}{r_{\rm c}}\right)- C+1\right),
\end{align}
where $C$ denotes the Euler constant.
\end{lemma}
\begin{proof}
See Appendix \ref{proof}.
\end{proof}
{\it Remark 3:} Define $\mathbb{P}^g =\frac{1}{2}f\left(\frac{4\lambda d^2}{r_{\rm c}}\right)$, where $f(x) = x(-\log x+a)$ and $a=1-C$. It is straightforward to show that  $f(0)=0$, and   $f(x)$ is monotonically increasing for $0\leq x\leq e^{a-1}$. Therefore, if $\frac{\lambda d^2}{r_{\rm c}}\rightarrow 0$, the outage probability shown in \eqref{lemma3eq} indicates that pinching-antenna assisted multi-cell transmission is guaranteed to outperform the conventional-antenna system, even though the channel between ${\rm U}_0$ and its serving antenna is identical in the two systems. 

\section{A Two-Dimensional General Case}\label{section 4}
In this section, a general two-dimensional (2D) case will be focused on, where multiple interfering BSs, rather than a single interfering BS as in the previous section, are considered. To characterize the impact of the potentially infinite number of interfering BSs, the probability generating functional (PGFL) needs to be used, which motivates the use of the ergodic data rate as the metric\cite{Haenggi}:
\begin{align}
\mathbb{E}^0 =& \mathbb{E}\left\{R_0^{\rm PA}
\right\} 
  = \mathbb{E}\left\{
\log\left(
1+ \frac{  \rho R_0^{  \alpha \epsilon }   g_0}{I_0 +1}
\right) 
\right\}  \\\nonumber
=&  \mathbb{E}\left\{
\log\left(
  \rho R_0^{ \alpha \epsilon }   g_0 +I_0 +1
\right) 
\right\}  -\mathbb{E}\left\{
\log\left( I_0 +1
\right) 
\right\},
\end{align}
where $I_0=\sum^{\infty}_{i=1}  \rho R_i^{ \alpha \epsilon}  g_i$ and $\rho=\frac{P}{P_{\rm N}}$.

Similar to the previous section, we assume  that ${\rm U}_0$ is right underneath of ${
\rm BS}_0$, which is to ensure that ${\rm U}_0$'s link quality to ${\rm BS}_0$ remains unchanged with or without using  pinching antennas.  Because ${\rm U}_0$ is very close to ${\rm BS}_0$ , it is assumed that $g_0=\frac{\eta}{d^2}$.
By using the identy: $\log(1+x) = \int^{\infty}_{0}\frac{1-e^{-tx}}{t}e^{-t}dt$, the ergodic data rate can be expressed as follows: 
\begin{align}\nonumber
\mathbb{E}^0 =&    \mathbb{E}\left\{  \int^{\infty}_{0}\frac{1-e^{-t\left( \rho \eta d^{ 2( \epsilon-1) }    +I_0\right)}}{t}e^{-t}dt
\right\}  \\\nonumber &-\mathbb{E}\left\{ \int^{\infty}_{0}\frac{1-e^{-tI_0}}{t}e^{-t}dt
\right\}  \\\label{er0} 
=&    \int^{\infty}_{0}\frac{\left(1-e^{-t  \rho\eta d^{ 2( \epsilon-1) }    } \right)\mathcal{L}_I(t)  }{t}e^{-t}dt ,
\end{align}
where the Laplace transform of the interference is defined as follows:  
\begin{align}
\mathcal{L}_I(s) =& \mathbb{E}\left\{  e^{-s I_0}
\right\}= \mathbb{E}\left\{  e^{-s \sum^{\infty}_{i=1}  \rho R_i^{ \alpha \epsilon }   g_i}
\right\} 
\\\nonumber
=&\mathbb{E}_{\tilde{\Phi}} \left\{\prod^{\infty}_{i=1} \mathbb{E}_{h_i,R_i}\left\{  e^{-s   \rho R_i^{ \alpha \epsilon}  g_i}
\right\} \right\}.
\end{align}

By applying the PGFL, the addressed Laplace function can be expressed as follows: 
\begin{align}\label{lapl}
\mathcal{L}_I(s) =&{\rm exp}
\left( -2 \pi \lambda \int_{0}^\infty \left(1 -  f_{r,s}(r,s) \right)rdr
\right),
\end{align}
where $f_{r,s}(r,s)$ is defined as follows:
\begin{align}
f_{r,s}(r,s) =& \mathbb{E}_{h_i,R_i}\left\{  e^{-s   \rho R_i^{ \alpha \epsilon }  g_i}
\right\}\\\nonumber  
=&  \mathbb{E}_{R_i}\left\{ e^{-\beta  R_i}e^{-\beta  \sqrt{r^2+d^2}}  e^{-s   \rho R_i^{2\epsilon} \frac{\eta}{r^2+d^2}}
\right\}. 
\end{align}
We note that the interference link from ${\rm BS}_i$ to ${\rm U}_0$ might suffer from an LoS loss, which can also happen to the link between ${\rm BS}_i$ to ${\rm U}_i$. Therefore, $f_{r,s}(r,s)$ can be evaluated as follows:
\begin{align}\label{frsrs}
f_{r,s}(r,s)  
=&  \mathbb{E}_{R_i}\left\{ e^{-\beta  R_i}e^{-\beta  \sqrt{r^2+d^2}}  e^{-s   \rho R_i^{2\epsilon} \frac{\eta}{r^2+d^2}}
\right\}\\\nonumber
&+\mathbb{E}_{R_i}\left\{ (1-e^{-\beta  R_i})   e^{-\beta  \sqrt{r^2+d^2}} e^{-s   \rho R_i^{4\epsilon} \frac{\eta}{r^2+d^2}}
\right\}
\\\nonumber &+  \mathbb{E}_{R_i}\left\{\frac{e^{-\beta  R_i}(1-e^{-\beta  \sqrt{r^2+d^2}}) }{1+s   \rho R_i^{2\epsilon}  \frac{\eta  }{(r^2+d^2)^2}}\right\}
\\\nonumber &+  \mathbb{E}_{R_i}\left\{\frac{(1-e^{-\beta  R_i})(1-e^{-\beta  \sqrt{r^2+d^2}})}{1+s   \rho R_i^{4\epsilon}  \frac{\eta  }{(r^2+d^2)^2}}\right\}. 
\end{align}

The pdf of $R_i$ can be obtained as follows. Recall that $R_i$ denotes the distance between the user and its corresponding transmitter. Consider a circle with radius  $r_c$ and its center at the origin, and there is a user uniformly distributed within the circle, whose coordinate is denoted by $(\tilde{x}_i, \tilde{y}_i)$, i.e., $\tilde{x}_i^2+ \tilde{y}_i^2<r_c$.  Since pinching antennas can be moved to the proximity of their users, $R_i=\tilde{y}_i$. Since the user is uniformally distributed, the joint pdf of $\tilde{x}_i$ and $ \tilde{y}_i$ is given by $f_{\tilde{x}_i, \tilde{y}_i}(\tilde{x}_i, \tilde{y}_i)=\frac{1}{\pi r_c^2}$, and hence the marginal pdf of $\tilde{y}_i$ (or $R_i$) is given by
\begin{align}\label{pdfri}
f_{\tilde{y}_i}(y) =&  \int^{\sqrt{r_c^2-y^2}}_{-\sqrt{r_c^2-y^2}}f_{\tilde{x}_i, \tilde{y}_i}(\tilde{x}, \tilde{y})d\tilde{x}
=\frac{ 2\sqrt{r_c^2-y^2} }{\pi r_c^2},
\end{align}
for $-r_c\leq y\leq r_c$. By combining \eqref{er0}, \eqref{lapl}, \eqref{frsrs}, and \eqref{pdfri}, the ergodic data rate achieved by  multi-cell pinching-antenna transmission can be obtained. 
 
{\it Remark 4:} Unlike the analytical results shown in the previous section, it is challenging to obtain a closed-form expression for the considered 2D case. In particular, the use of the ergodic data rate as the metric facilitates the application of PGFL. However, the resulting integral shown in \eqref{er0} generally does not admit a closed-form analytical expression, particularly given the fact that the expression of $\mathcal{L}_I(s)$ is also complex for the considered case.


\vspace{-1em}
 \section{Simulation Results}\label{section 5}
For all conducted computer simulations,  the carrier frequency is $f_c=28$ GHz, $d=3$ m, $\epsilon=1$, $\beta=0.5$, and the noise power is $-90$ dBm.
 
 In Fig. \ref{fig1}, the special case studied in Section \ref{section 3} is focused on, i.e., the BSs follow a 1D HPPP. In particular, Fig. \ref{fig1a} focuses on the case with $R_t=2$ bits per channel use (BPCU) and $P=0$ dBm. As can be seen from the figure, regardless of the choices of $\lambda$ and $r_c$, pinching-antenna-based transmission always outperforms the benchmarking scheme. It is also worth pointing out that by increasing $r_c$ from $10$ m to $30$ m, the outage performance of the pinching-antenna scheme is improved, since the interfering transmitters are farther from the typical user. For the benchmarking scheme, increasing $r_c$ increases the distance between a BS and its serving user, which leads to performance degradation. Fig. \ref{fig1b} further confirms the superior performance gain  of pinching antennas over conventional antennas, and also shows that the pinching-antenna scheme is robust to the change of the transmit power. Fig. \ref{fig1} also verifies the accuracy of the analytical results shown in Lemma \ref{lemma1}.

   \begin{figure}[!] \vspace{-0.2em}
\begin{center}
\subfigure[ $R_t=2$ BPCU and $P=0$ dBm ]{\label{fig1a}\includegraphics[width=0.35\textwidth]{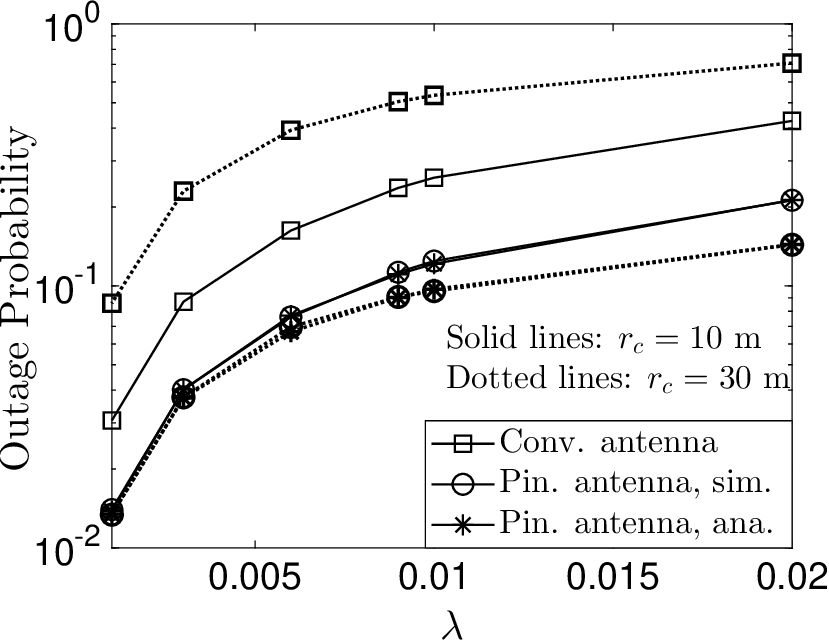}} 
\subfigure[$R_t=4$ BPCU and $r_c=30$ m]{\label{fig1b}\includegraphics[width=0.35\textwidth]{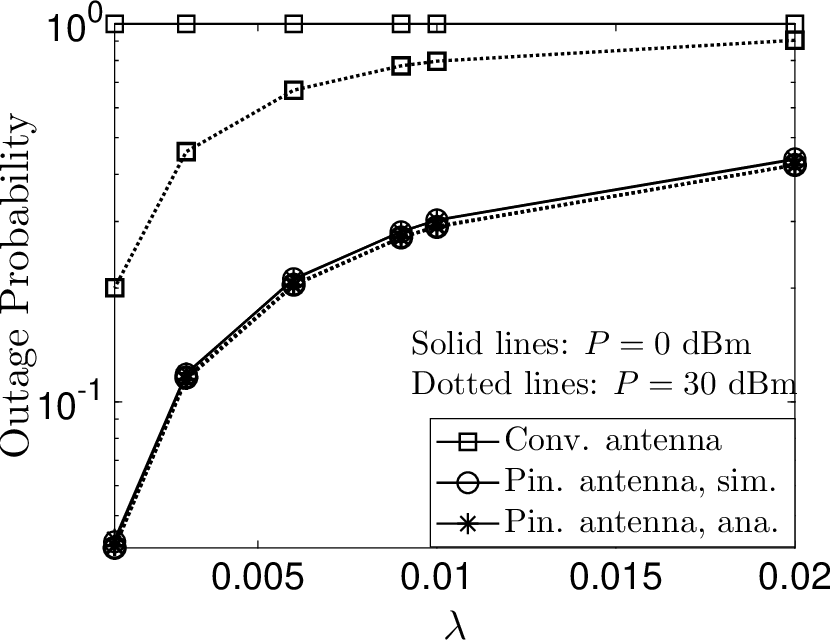}}   \vspace{-1em}
\end{center}
\caption{Outage performance of multi-cell pinching-antenna transmission ($\mathbb{P}^o$), where the BSs follow a 1D HPPP.       \vspace{-1em} }\label{fig1}\vspace{-1em}
\end{figure}

In Fig. \ref{fig2}, the probability of conventional antennas outperforming pinching antennas  ($\mathbb{P}^g$) is illustrated. As indicated in Lemmas \ref{lemma2} and \ref{lemma3}, when $\lambda$ is small, pinching-antenna assisted multi-cell transmission is more likely to outperform the conventional-antenna system by decreasing $\lambda$ and increasing $r_c$. This prediction is confirmed by Fig. \ref{fig2}. In addition, Fig. \ref{fig2} verifies the accuracy of the analytical results shown in Lemma \ref{lemma2}, and Table \ref{table1} shows the accuracy of the approximated analytical results shown in Lemma \ref{lemma3}. Fig. \ref{fig2} also shows that for moderate or large values of $\lambda$, $\mathbb{P}^g$ is a concave function of $\lambda$, instead of the monotonically increasing function, where the results shown in Lemma \ref{lemma3} are valid in the small-$\lambda$ regime. 
      \begin{figure}[t]\centering \vspace{-0.2em}
    \epsfig{file=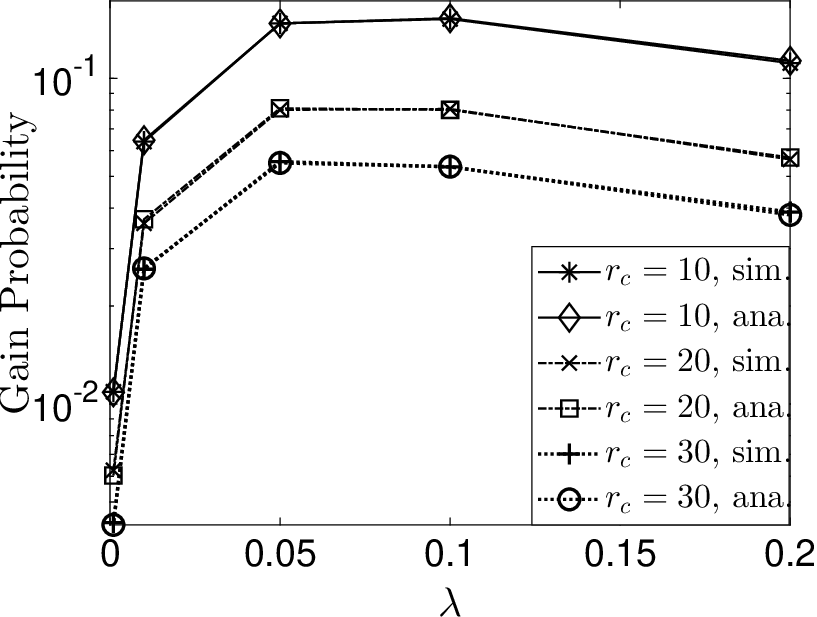, width=0.35\textwidth, clip=}\vspace{-0.5em}
\caption{ The probability of conventional antennas outperforming pinching antennas  ($\mathbb{P}^g$), where the BSs follow a 1D HPPP and $P=0$ dBm. Fixed transmit power, rather than fractional power control, is used. 
  \vspace{-1em}    }\label{fig2}   \vspace{-0.1em} 
\end{figure}
 \begin{table}[!]
\centering
\caption{The probability of conventional antennas outperforming pinching antennas  ($\mathbb{P}^g$) \vspace{-1em}}
\begin{tabular}{*7c}
\toprule 
$\lambda$&    $0.001$&$0.002$  & $0.003$  & $0.004$ &$0.005$ &$0.006$   \\
    \hline
 Sim.   &$0.0111$    &$0.0201$   &$ 0.0271$   &$ 0.0339$   &$ 0.0398$   &$ 0.0443$
    \\
    \hline
 Ana.    &$ 0.0109$   &$ 0.0192 $  &$ 0.0265$   &$ 0.0331$    &$0.0392$   &$ 0.0449$
  \\
    \hline
  App.   &$    0.0109$&$    0.0193$&$    0.0267$  &$  0.0336$ &$   0.0400$ &$    0.0460   $  \\
\bottomrule
\end{tabular}\label{table1}\vspace{-1em}
\end{table}

   \begin{figure}[!] \vspace{-0.2em}
\begin{center}
\subfigure[ Typical user's location is random ]{\label{fig3a}\includegraphics[width=0.35\textwidth]{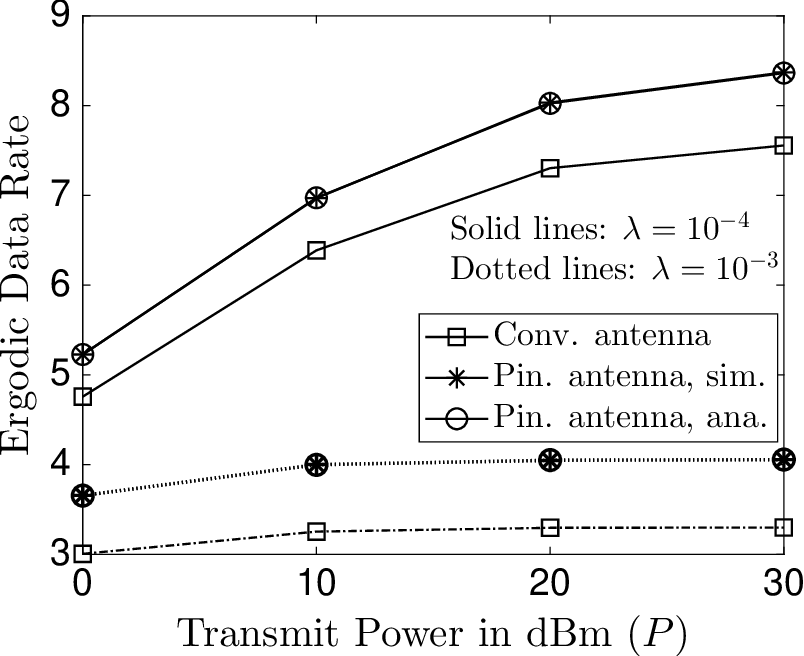}} 
\subfigure[Typical user is underneath of its base station]{\label{fig3b}\includegraphics[width=0.35\textwidth]{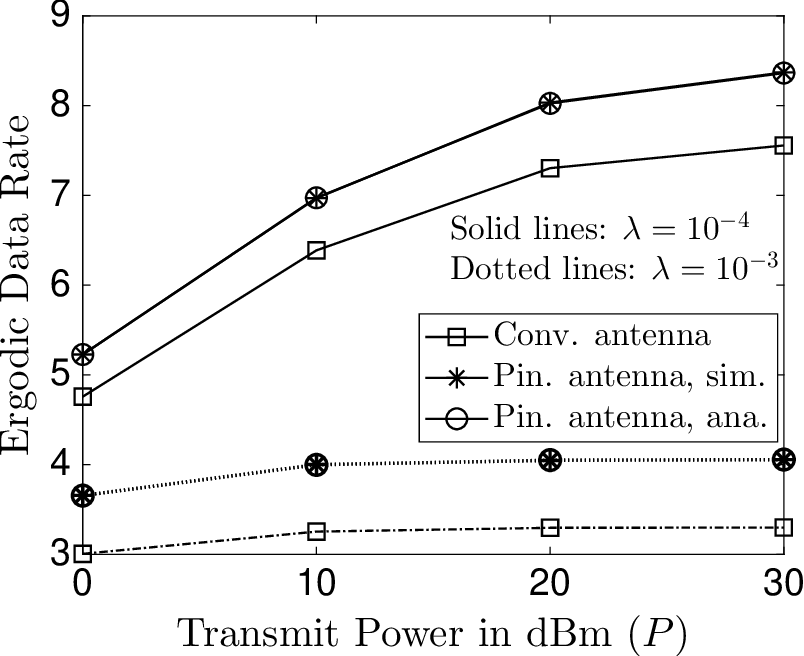}}   \vspace{-1em}
\end{center}
\caption{Ergodic data rates achieved by pinching-antenna-assisted multi-cell transmission, where the BSs follow a 2D HPPP and   $r_c=30$ m.     \vspace{-1em} }\label{fig3}\vspace{-1em}
\end{figure}

\begin{figure*}\vspace{-2em}
\begin{align}\label{eqall}
P=&
\frac{1-e^{-2\lambda r_{\rm c}}}
{4\lambda r_{\rm c}}
-
\frac{\sqrt{2}d}{ r_{\rm c}}
\left[
K_1\left(4\sqrt{2}\lambda d\right)
 -\frac{r_{\rm c}}{2\sqrt{2}d}
\sum^{\infty}_{j=0}
 \left(2\lambda r_{\rm c}\right)^{j-1}\Gamma\left(
1-j,2\lambda r_{\rm c}
\right)\frac{\left(-\frac{4\lambda d^2}{r_{\rm c}}\right)^j}{j!}
\right]
\\\nonumber
\approx&
\frac{1-e^{-2\lambda r_{\rm c}}}
{\lambda r_{\rm c}}
-
\frac{\sqrt{2}d}{r_{\rm c}}
\left[
\frac{1}{4\sqrt{2}\lambda d} +2\sqrt{2}\lambda d\left(\log \left(2\sqrt{2}\lambda d\right)+C-\frac{1}{2}\right)
- \frac{r_{\rm c}}{2\sqrt{2}d}
\left(
  \frac{e^{-2\lambda r_{\rm c}}}{2\lambda r_{\rm c}}  +C \frac{4\lambda d^2}{r_{\rm c}}  + \frac{4\lambda d^2}{r_{\rm c}}\log(2\lambda r_{\rm c})   
\right)
\right] 
\\\nonumber 
\approx&
\frac{2\lambda d^2 }{r_{\rm c}} \left(-2\log \left(2\sqrt{2}\lambda d\right)- C+1
    +  \log(2\lambda r_{\rm c})\right)
=
\frac{2\lambda d^2 }{r_{\rm c}} \left(- \log \left(\frac{4\lambda d^2}{r_{\rm c}}\right)- C+1\right).
\end{align}\vspace{-2em}
\end{figure*}

In Fig. \ref{fig3}, the ergodic data rate achieved by multi-cell pinching-antenna transmission is illustrated, where the case considered in Section \ref{section 4} is focused on, i.e., the BSs follow a 2D HPPP. In particular, Fig. \ref{fig3a} focuses on the case where users are randomly distributed in their cells. As can be seen from the figure, the use of pinching antennas yields a significant performance gain over conventional antennas, particularly at high SNR.  This performance gain is due to the following facts. First, the use of pinching antennas improves the link between ${\rm BS}_0$ and ${\rm U}_0$. Second, the use of pinching antennas reduces the interfering BSs' transmit powers and hence suppresses co-channel interference. In Fig. \ref{fig3b}, the typical user is assumed to be underneath its BS, which means that the use of pinching antennas does not improve the typical user's connection to its BS. Fig. \ref{fig3b} shows that even in such a case, the use of pinching antennas can still bring a significant performance gain over conventional antennas, due to the fact that the multi-cell interference is effectively suppressed by using pinching antennas. 

\section{Conclusions}
This letter focused on an important benefit of using pinching antennas to suppress multi-cell interference.  A stochastic geometric study has been carried out in the letter, where  BSs employ fractional power control. Analytical and numerical results have been presented to demonstrate the superior performance gain of multi-cell pinching-antenna transmission over the conventional scheme.
\appendices
\section{Proof for Lemma \ref{lemma3}}\label{proof}

To facilitate the approximation, the probability for conventional antennas to outperform pinching antennas shown in Lemma \ref{lemma2} can be first expressed as follows\cite{besselsss}: 
\begin{align}
\mathbb{P}^g=&
\frac{1-e^{-2\lambda r_{\rm c}}}
{4\lambda r_{\rm c}}
-
\frac{\sqrt{2}d}{r_{\rm c}}
\left[
K_{1}\left(4\sqrt{2}\lambda d\right)
\right.\\\nonumber &\left.-\frac{r_{\rm c}}{2\sqrt{2}d}
\sum^{\infty}_{j=0}
\underset{\xi_j}{ 
\underbrace{\left(2\lambda r_{\rm c}\right)^{j-1}\Gamma\left(
1-j,2\lambda r_{\rm c}
\right)\frac{\left(-\frac{4\lambda d^2}{r_{\rm c}}\right)^j}{j!}}}
\right].
\end{align}

Recall that the two functions, $K_1(z)$ and $\Gamma(1-j,z)$, can be approximated as follows\cite{GRADSHTEYN}:
\begin{align}
&K_1(z) \approx \frac{1}{z} +\frac{z}{2}\left(\log \frac{z}{2}+C-\frac{1}{2}\right),\\\nonumber
&\Gamma(1-j,z) \approx  \frac{(-1)^{j-1}}{(j-1)!}\left[
-C-\log(z)- e^{-z}\sum^{j-2}_{m=0}(-1)^m\frac{m!}{z^{m+1}}
\right],
\end{align}
for $j\geq 1$ and $z\rightarrow 0$.

For $j=0$, $\xi_j$ can be expressed as follows:
\begin{align} \label{xi1}
 \xi_0=\left(2\lambda r_{\rm c}\right)^{-1}\Gamma\left(
1,2\lambda r_{\rm c}
\right) =  \left(2\lambda r_{\rm c}\right)^{-1}e^{-2\lambda r_{\rm c}}. 
\end{align}

To evaluate $\xi_j$, $j\geq 1$, we first note the following approximation: 
\begin{align} 
&\psi_j \triangleq z^{j-1}\Gamma(1-j,z) \approx \frac{(-1)^{j-1}}{(j-1)!}\\\nonumber
 & \times \left[-C z^{j-1}
 -z^{j-1}\log(z)- e^{-z}\sum^{j-2}_{m=0}(-1)^m m! z^{j-m-2} 
\right] .
\end{align}

For $j=1$, the term $\psi_j $ can be approximated as follows:
\begin{align} 
 \psi_1= \Gamma(0,z)  
 \approx&    -C   - \log(z)  . 
\end{align}

For $j\geq 2$, the term $\psi_j $ can be approximated as follows:
\begin{align} 
\psi_j   \approx& \frac{(-1)^{j-1}}{(j-1)!}\left[
 -z^{j-1}\log(z)- e^{-z} (-1)^{j-2} (j-2)!  
\right].
\end{align}

Therefore, the overall approximation for $\sum^{\infty}_{j=0}\xi_j$ can be obtained by utilizing the following approximation with small $z$: 
\begin{align} 
&\sum^{\infty}_{j=0}
 \left(z\right)^{j-1}\Gamma\left(
1-j,z
\right)\frac{a^j}{j!} \\\nonumber
\approx & \frac{e^{-z}}{z}  -C a  - a\log(z) + \sum^{\infty}_{j=2}
 \left(z\right)^{j-1}\Gamma\left(
1-j,z
\right)\frac{a^j}{j!} 
\\\nonumber  
\approx & \frac{e^{-z}}{z}  -C a  - a\log(z) - \log(z)\sum^{\infty}_{j=2} \frac{(-1)^{j-1}a^j}{(j-1)!j!}z^{j-1}   
  \\\label{xi2}
&
+\sum^{\infty}_{j=2} \frac{1}{(j-1)} e^{-z}    
  \frac{a^j}{j!} 
\approx   \frac{e^{-z}}{z}  -C a  - a\log(z)   . 
\end{align}
 By applying the approximation shown in \eqref{xi2},   the probability for conventional antennas to outperform pinching antennas shown in Lemma \ref{lemma3} can be approximated as shown in \eqref{eqall} at the top of this page, and hence the lemma is proved.   

%

  \vspace{-0.5em}
\bibliographystyle{IEEEtran}
\bibliography{IEEEfull,trasfer}
  \end{document}